\documentclass[11pt,a4paper]{article}
\usepackage{amssymb,amsmath,amsthm,color}
\usepackage[a4paper,margin=2.5cm,footskip=1cm,headsep=1.5cm]{geometry}
\usepackage{graphicx}
\def\ZR{{\mathbb R}}

\def\ZN{{\mathbb N}}

\def\beq{\begin{equation}}
\def\eeq{\end{equation}}
\def\be{\begin{equation}}
\def\ee{\end{equation}}
\def\beqar{\begin{eqnarray}}
\def\eeqar{\end{eqnarray}}
\def\ber{\begin{eqnarray}}
\def\eer{\end{eqnarray}}
\def\berb{\begin{eqnarray*}}
	\def\eerb{\end{eqnarray*}}

\def\Ker#1{\mathop{\rm Ker}\nolimits#1}

\def\const{{\rm const}}

\def\norm#1.#2.{\|#1\|_{#2}}
\def\Norm#1.#2.{\big\|#1\big\|_{#2}}
\def\NOrm#1.#2.{\bigg\|#1\bigg\|_{#2}}
\def\NORm#1.#2.{\Big\|#1\Big\|_{#2}}
\def\NORM#1.#2.{\Bigg\|#1\Bigg\|_{#2}}

\def\vec#1{{\mathchoice{\mbox{\boldmath$\displaystyle#1$}}
		{\mbox{\boldmath$\textstyle#1$}}
		{\mbox{\boldmath$\scriptstyle#1$}}
		{\mbox{\boldmath$\scriptscriptstyle#1$}}}}

\def \0b{{\hbox{\boldmath $0$}}}

 \def \b{\vec{b}}

\newcommand{\eb}{\vec{e}} \newcommand{\fb}{\vec{f}}

\newcommand{\mb}{\vec{m}} \newcommand{\nb}{\vec{n}}

 \newcommand{\tb}{\vec{t}}
\newcommand{\ub}{\vec{u}} \newcommand{\vb}{\vec{v}}
 \newcommand{\xb}{\vec{x}}
\newcommand{\yb}{\vec{y}} 

\newcommand{\Abb}{{\bf A}} \newcommand{\Bbbb}{{\bf B}}

 \newcommand{\Hbb}{{\bf H}}
\newcommand{\Ibb}{{\bf I}} 
 
\newcommand{\Mbb}{{\bf M}} \newcommand{\Nbb}{{\bf N}}
 
\newcommand{\Qbb}{{\bf Q}}

\newcommand{\Mb}{\vec{M}} \newcommand{\Nb}{\vec{N}}
 
\newcommand{\Qb}{\vec{Q}} 
 
\newcommand{\Ub}{\vec{U}} \newcommand{\Vb}{\vec{V}}
 
\newcommand{\Yb}{\vec{Y}} 

\def \alb{\vec{\alpha}} \def \betab{\vec{\beta}}

\def \thb{\vec{\theta}} 
 
\def \lab{\vec{\lambda}}

 \def \phib{\vec{\phi}}
 \def \chib{\vec{\chi}}

\def \Thetab{\vec{\Theta}} \def \Lab{\vec{\Lambda}}

\def \Phib{\vec{\Phi}} 
\def \Omegab{\vec{\Omega}}

\newcommand{\cI}{{\cal I}} 
 
 \newcommand{\cN}{{\cal N}}

\newcommand{\cS}{{\cal S}} 
 \newcommand{\cV}{{\cal V}}

\def \ntb{\vec{\tilde{n}}}

\def \utb{\vec{\tilde{u}}}

\def \ytb{\vec{\tilde{y}}}

\def \yhb{\vec{\hat{y}}}

\newtheorem{remark}{Remark}[section]

\newcounter{primjer}[section]
\newcounter{tvrdnja}[section]
\newcounter{zadatak}[section]
\newcommand{\Pbbb}{{\mathbb P}}
\newcommand{\Phitb}{{\vec{\tilde{\Phi}}}}
\newcommand{\Mbbbb}{\Mbb^{6\times6}}
\newcommand{\sspan}{\mathop{\rm span}}
\renewcommand{\Im}{\mathop{\rm Im}}

\DeclareMathOperator{\Ran}{Im}

\newtheorem{theorem}{Theorem}[section]
\newtheorem{lemma}[theorem]{Lemma}
\newtheorem{proposition}[theorem]{Proposition}
\newtheorem{definition}{Definition}[section]

\newcommand{\thetab}{\boldsymbol{\theta}}

\newcommand{\Vspace}{H^1(0,\ell;\ZR^3)\times H^1(0,\ell;\ZR^3)}

\newcommand{\thetatb}{\vec{\tilde\theta}}
\def \Omegab{\vec{\Omega}}

\def \thhhb{\vec{\hat{\theta}}}

\newcommand{\Vstent}{V_{S0}}

\title{Mixed formulation of the one-dimensional equilibrium model for
elastic stents}

\author{Luka Grubi\v{s}i\'{c}\footnote{Department of Mathematics, University of Zagreb, Croatia, luka.grubisic@math.hr}, Josip Ivekovi\'{c}\footnote{jivekovi@gmail.com}, Josip Tamba\v{c}a\footnote{Department of Mathematics, University of Zagreb, Croatia, tambaca@math.hr} \ and Bojan \v{Z}ugec\footnote{Faculty of Organization and Informatics, University of Zagreb, Vara\v{z}din, Croatia, bojan.zugec@foi.hr}}
\date{\today}

\begin{document}

\maketitle

\begin{abstract}
In this paper we formulate and analyze the mixed formulation of the one-dimensional equilibrium model of elastic stents. The model is based on the curved rod model for the inextensible and ushearable struts and is formulated in the weak form in \cite{IMASJ}. It is given by a system of ordinary differential equations at the graph structure. In order to numerically treat the model using finite element method the mixed formulation is plead for. We obtain equivalence of the weak and the mixed formulation by proving the Babuska--Brezzi condition for the stent structure.
\end{abstract}

\tableofcontents

\section{Introduction}

A stent is a mesh tube that is inserted into a natural conduit of the body to prevent or counteract a disease-induced localized flow constriction. For instance, to keep the arteries open, a stent is inserted at the location of the narrowing. Performance of coronary stents depends on the mechanical properties of the material the stent is made of and on the geometrical properties of a stent, e.g., number of stent struts, their length, their placement, the strut width and thickness,  geometry of the cross section of each stent strut, etc. As a consequence the behavior of the stents is very complex and reliable models are desirable.

Since the stents are usually made of metals we consider a stent to be a three-dimensional elastic body
defined as a union of three-dimensional struts.
If the deformations in the problem are small behavior of stents can be modeled by the linearized elasticity.
The equations of linearized elasticity in thin domains are very demanding for numerical approximation and qualitative analysis. Therefore, it is appealing to construct a more simple analytical approximation. The one-dimensional model of stents was first considered in \cite{SIAMstent} and then reformulated in \cite{IMASJ} as a system of ordinary differential equations given on the graph defined by the middle curves of the stent struts. The model is one-dimensional in a sense that it is given by the ordinary differential equations with respect to the natural parameter of the middle curves of the struts as the variable (arc-length variable). However the model describes full three-dimensional behavior of the stent. Thus, each stent strut is modeled using the curved rod model (see \cite{JT1,JT2} for the rigorous derivation of the model from three-dimensional elasticity) and a set of transmission conditions at joints (vertices of the graph) describing: continuity of displacement and rotation and equilibrium of forces and couples.
Note that the model is not restricted to stents but can be used to model behavior of any elastic structure made of elastic bodies which are thin in two directions (rod--like).

The function space on which the stent model is posed in the weak formulation, includes conditions of inextensibility and ushearability of rods that model struts. Therefore to build a finite element approximation within this function space ($\Vstent$ in Section~\ref{Sstent}) one needs to fulfill these restrictions with each finite element functions. However this is not a simple task. The associated mixed formulation removes the conditions from the function space using the Lagrange multipliers. These conditions then become adjoined equations.

When the stent is considered solely with no interaction with surrounding, natural boundary conditions are of the Neumann type, i.e., contact forces are prescribed. As usual associated to this pure traction problems are two qualitative properties, a necessary condition for the existence (total force and moment are zero, see  (\ref{necessary})) and nonuniqueness of the solution of the problem (up to an infinitesimal rigid deformation, see (\ref{uniqueness})). We remove both notions by fixing the total displacement and total infinitesimal rotation in (\ref{uniqueness}). This adds two more equations in the mixed formulation (and two more Lagrange multipliers which are in $\ZR^3$).

The main tool to obtain the equivalence of the classical weak (variational) formulation and the mixed formulation is the Babuska-Brezzi condition, see \cite[Theorem~4.1 and Corollary 4.1]{GR}
or \cite[Theorem II.1.1]{BrezziFortin}. To obtain it we use techniques applied in \cite{Zugec0} and \cite{Zugec1} for solving a particular problem on the same graph prescribed by the stent, but for simplified constitutive law, zero forces and prescribed non-zero extension and shear of the cross--section. The  solution is obtained by explicitly integrating problems on struts, and then incorporating the solution, with respect to the topology of the graph, into a big algebraic system. The numerical analysis of the discretization together with a description of an implementation of the method will be presented in \cite{GruTamIv}.

In Section~\ref{S2} we start with the problem on a single curved rod and the formulation of its mixed formulation to introduce the problem and present the ideas. Then, in Section~\ref{S3} we formulate the stent model, recall the $H^1$ space on the graph, and prove the Poincar\'{e} type inequality for this space. Using this estimate we prove the ellipticity of the stiffness form on the space that includes inextensibility and unshearability of rods. Finally we prove the Babuska-Brezzi $\inf-\sup$ condition for the stent structure in Lemma~\ref{infsup} which then gives the equivalence of the formulations.

\section{1D curved rod model\label{S2}}
\setcounter{equation}{0}

\subsection{Differential formulation}
A three-dimensional elastic body with its two dimensions small comparing to the third is generally called an elastic rod, see Figure~\ref{strut}.
\begin{figure}[ht]
\begin{center}
\includegraphics[width=0.45\textwidth]{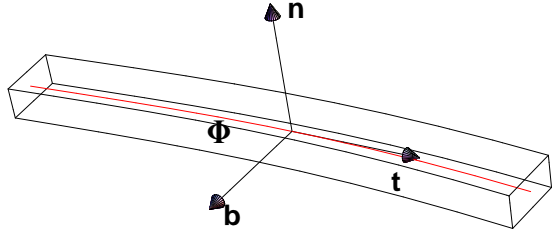}
\end{center}
\caption{3D thin elastic body\label{strut}}
\end{figure}
A curved rod model is a one-dimensional
approximation of a "thin" three-dimensional curved elastic structure
given in terms of the arc-length of the middle curve of the rod as
an unknown variable. Thus in order to build the model a natural parametrization $\Phib:[0,l] \to \ZR^3$ of the middle curve of the curved rod
(red in Figure~\ref{strut}) has to be given.
Further, let the cross-section of a rod be rectangular, of width $w$ in direction of the binormal $\b$ on the middle curve and thickness $t$ in direction of the normal $\nb$ on the middle curve.

One-dimensional equilibrium model for curved elastic rods we use here is given by the following first order system. For a given force with line
density $\fb$ the model is expressed in terms of $(\yb,\thetab,\mb,\nb)$ that satisfy
\begin{align}
\label{lj1r} 0 &= \partial_s \nb + \fb,\\
\label{lj2r} 0 &= \partial_s \mb  +  \tb \times\nb%
,\\
\label{lj3r} 0 &= \partial_s \thetab - \Qbb \Hbb^{-1} \Qbb^T \mb,\\
\label{lj4r}
0 &= \partial_s \yb +\tb \times  \thetab
\end{align}
together with associated boundary condition. Here, $\yb$ is the displacement of the middle curve, $\thetab$ is the vector of the infinitesimal rotation of the cross-sections, $\mb$ is the contact couple and $\nb$ is the contact force. The first two equations describe the balance of contact force and contact moment, respectively,
while the last two equations describe the constitutive relation for a curved, linearly elastic rod.
The last equation can be interpreted as the
condition of inextensibility and unshearability of the rod.
The matrices $\Hbb$ and $\Qbb$ are given by
$$
\Hbb = \left[\begin{array}{ccc}
\mu K & 0 & 0 \\
0 &E I_{11} & E I_{12} \\
0 & E I_{12} & E I_{22} \\
\end{array}\right], \qquad \Qbb = \left[\begin{array}{ccc} \tb & \nb & \b\end{array}\right]
$$
see e.g. \cite{CaoTucker}.
Here  $E = \mu (3\lambda + 2\mu)/(\lambda+\mu)$ is the Youngs modulus of the material
($\mu$ and $\lambda$ are the Lam\'e constants), $I_{ij}$ are the moments of inertia of the
cross-sections and $\mu K$ is the torsion rigidity of the cross-sections.
Therefore, $\Hbb$ describes the elastic properties of the material the rods (struts) are made of
and the geometry of the cross-sections.

This model is linearization of the Antman-Cosserat model for inextensible, unshearable rods,
see \cite{Antman} for the nonlinear model and \cite{IMASJ} for the linearization. The model can also be seen as a linearization of the nonlinear model derived by Scardia in \cite{Scardia} from three-dimensional nonlinear elasticity.
It was show in \cite{JT1} and \cite{JT2} that the solution of the one-dimensional model can be obtained as a limit of solutions of equilibrium equations of three-dimensional elasticity
when thickness of the cross-sections (both, $w$ and $t$) tend to zero. Corresponding result for the dynamic case is given in \cite{Tevolution}. Therefore, for three-dimensional rods which are thin enough one-dimensional curved rod model can provide well enough approximation.
Moreover, in \cite{T1}, it was shown that curved geometry can be approximated with a piecewise straight geometry with an error estimate.
This will further simplify the equations of the one-dimensional model.


\subsection{Mixed and weak formulations}

To obtain classical existence and uniqueness result one classically rewrites the problem in the weak/variational formulation. In the weak formulation inextensibility and unshearability conditions are incorporated in the function space which complicates the numerical approximation of the problem. Therefore the mixed formulation of (\ref{lj1r})-(\ref{lj4r}) is called for.
%

Let us take $(\ytb,\thetatb)\in V = \Vspace$ and multiply (\ref{lj1r}) by $\ytb$ and (\ref{lj2r}) by $\thetatb$ and sum the equations and
integrate them over $[0,\ell]$. We obtain
$$
\aligned
0 = \int_0^\ell \partial_s \nb \cdot \ytb ds + \int_0^\ell \fb \cdot \ytb ds
+\int_0^\ell \partial_s \mb \cdot  \thetatb ds +  \int_0^\ell \tb \times\nb \cdot \thetatb ds.
\endaligned
$$
After partial integration in the first and the third term on the right hand side we obtain
$$
\aligned
0 =&\ - \int_0^\ell  \nb \cdot \partial_s \ytb ds + \int_0^\ell \fb \cdot \ytb ds
- \int_0^\ell \mb \cdot \partial_s  \thetatb ds +  \int_0^\ell \thetatb \times \tb \cdot\nb   ds\\
&\ + \nb(\ell) \cdot  \ytb(\ell) - \nb(0) \cdot  \ytb(0) + \mb(\ell) \cdot   \thetatb(\ell) - \mb(0) \cdot   \thetatb(0).
\endaligned
$$
Inserting  $\mb$  from (\ref{lj3r}) we obtain
$$
\aligned
0 =&\ - \int_0^\ell  \nb \cdot (\partial_s \ytb + \tb \times \thetatb ) ds + \int_0^\ell \fb \cdot \ytb ds
- \int_0^\ell \Qbb \Hbb \Qbb^T \partial_s \thetab  \cdot \partial_s  \thetatb ds \\
&\ + \nb(\ell) \cdot  \ytb(\ell) - \nb(0) \cdot  \ytb(0) + \mb(\ell) \cdot   \thetatb(\ell) - \mb(0) \cdot   \thetatb(0).
\endaligned
$$
Finally from (\ref{lj4r}) for all $\ntb \in L^2(0,\ell; \ZR^3)$ we obtain
$$
\int_0^\ell \ntb \cdot (\partial_s \yb +\tb \times  \thetab) ds =0.
$$
Let us define function spaces $V= \Vspace$, $Q=L^2(0,\ell; \ZR^3)$, bilinear forms
$$
\aligned
&k: V \times V \to \ZR, \qquad k((\yb,\thetab),(\ytb,\thetatb))=\int_0^\ell \Qbb \Hbb \Qbb^T \partial_s \thetab  \cdot \partial_s  \thetatb ds,\\
&b: Q \times V \to \ZR, \qquad b(\nb,(\ytb,\thetatb))=\int_0^\ell  \nb \cdot (\partial_s \ytb + \tb \times \thetatb ) ds,
\endaligned
$$
and the linear functional
$$
l:V \to \ZR, \qquad l(\ytb,\thetatb) = \int_0^\ell
\fb \cdot \ytb ds.
$$
Now the mixed formulation of the one rod problem (\ref{lj1r})-(\ref{lj4r}) is given by: find $((\yb,\thetab),\nb) \in V \times Q$ such that
\begin{equation}\label{weak_one1}
\aligned
&k((\yb,\thetab),(\ytb,\thetatb)) + b(\nb,(\ytb,\thetatb) )\\
&\qquad = l(\ytb,\thetatb) + \mb(\ell)\cdot \thetatb(\ell) - \mb(0)\cdot  \thetatb(0) + \nb(\ell)\cdot  \ytb(\ell) - \nb(0)\cdot  \ytb(0), \quad (\ytb,\thetatb) \in V,\\
&b(\ntb,(\yb,\thetab))=0,\qquad \ntb \in Q.
\endaligned
\end{equation}

For a single rod, the boundary conditions at $s=0,\ell$ need to be prescribed. At this point we assume (the most difficult case) that the rod is clamped at $s=0$ and $s=\ell$, i.e.,
$$
\yb(0)=\thb(0)=\yb(\ell)=\thb(\ell)=0.
$$
Therefore we define the function space containing these boundary conditions, namely
$$
V^0 = \{(\yb,\thb)\in V : \yb(0) = \thb(0)=\yb(\ell)=\thb(\ell)=0\}.
$$
Now the mixed formulation is given by:
find $((\yb,\thetab),\nb) \in V^0 \times Q$ such that
\begin{equation}\label{weak_one}
\aligned
&k((\yb,\thetab),(\ytb,\thetatb)) + b(\nb,(\ytb,\thetatb) ) = l(\ytb,\thetatb) , \quad (\ytb,\thetatb) \in V^0,\\
&b(\ntb,(\yb,\thetab))=0,\qquad \ntb \in Q.
\endaligned
\end{equation}

\begin{remark}\label{rjunction}\em
For the stent problem, the boundary conditions will be given by the kinematic and dynamic contact
conditions. They consist of continuity of displacement and infinitesimal rotation and requirement that the sum of contact forces be equal to zero,
and that the sum of contact moments be equal to zero,
for all rods meeting at the given vertex.
\end{remark}

To the problem (\ref{weak_one}) we can also associate the weak formulation. For that we first define the subspace
$$
V^0_0 =\{(\yb,\thb)\in V^0 : \partial_s \yb + \tb \times \thetab =0\}
$$
of $V^0$ which includes the inextensibility and ushearability condition given by (\ref{lj4r}). Then the weak/variational formulation is given by:
find $(\yb,\thetab)\in V^0_0$ such that
\begin{equation}\label{weak_one_0}
k((\yb,\thetab),(\ytb,\thetatb)) = l(\ytb,\thetatb), \quad (\ytb,\thetatb) \in V^0_0.
\end{equation}

Note also that the form $b$ defines the linear operator $B: V \to Q'$ by
$$
b(\nb, (\ytb,\thetatb)) = {}_{Q'}\langle B (\ytb,\thetatb), \nb \rangle_{Q}.
$$
This operator is important for the analysis of the mixed formulation.

\begin{lemma}\label{lexistence_one}
One has:
\begin{itemize}
\item[a)] the form $k$ is $V^0_0$--elliptic.
\item[b)] $B (V^0)$ is closed in $Q'$.
\begin{itemize}
\item If the parametrization of the middle curve of the rod $\Phib$ is not affine then $\Im{B} = Q=L^2(0,\ell;\ZR^3)$.
\item If $\Phib$ is affine with constant tangent $\tb$ then
$$
B(V^0) = \{\lab \in L^2(0,\ell;\ZR^3): \int_0^\ell \lab (s) ds \cdot \tb =0\}.
$$
\end{itemize}
\end{itemize}
\end{lemma}
\begin{proof}
a) For $(\yb,\thetab)\in V^0_0$ we estimate using the Poincare inequality three times
$$
\aligned
\|(\yb,\thetab)\|_V^2 &\leq C( \|\yb'\|_{L^2(0,\ell;\ZR^3)}^2 + \|\partial_s \thetab\|_{L^2(0,\ell;\ZR^3)}^2)\\
&\leq C( \|\yb'+ \tb \times \thetab\|_{L^2(0,\ell;\ZR^3)}^2 + \|\tb \times \thetab\|_{L^2(0,\ell;\ZR^3)}^2 +  \|\partial_s \thetab\|_{L^2(0,\ell;\ZR^3)}^2)\\
&\leq C\|\partial_s  \thetab\|_{L^2(0,\ell;\ZR^3)}^2 = C\|\Qbb^T \partial_s \thetab\|_{L^2(0,\ell;\ZR^3)}^2\\
&\leq \frac{C}{\min{\sigma(\Hbb)}} k((\yb,\thetab),(\yb,\thetab)).
\endaligned
$$

b) Let us take $\lab \in Q$. We try to find $(\yb,\thb) \in V^0$ such that
$$
\lab = \yb'+ \tb \times \thb.
$$
For that we search for the solution of the system
\begin{equation}\label{system}
\aligned
&\nb'=0,\\
&\mb' + \tb \times \nb=0,\\
&\thb'-\mb=0,\\
&\yb' + \tb \times \thb =\lab
\endaligned
\end{equation}
such that  $(\yb,\thb) \in V^0$. Note that we actually need to solve only the last equation and that the first three equations are arbitrary. However this formulates the system very similar to the one already analyzed in \cite{Zugec0}. Therefore, for some constants $\Nb,\Mb \in \ZR^3$ one has
$$
\nb(x) = \Nb, \quad \mb(x) = \Mb - \int_0^x \tb(s) \times \Nb ds = \Mb - \Abb_{\Phib(x)-\Phib(0)} \Nb,
$$
where $\Abb_\vb$ is the skew-symmetric matrix associated with the vector $\vb$, i.e., $\Abb_\vb \xb = \vb \times \xb$.
Integrating the third equation in the system we obtain
$$
-\thb(x) = \thb(\ell)-\thb(x) = \int_x^\ell \mb(s) ds = \int_x^\ell \Mb - \Abb_{\Phib(s)-\Phib(0)} \Nb ds= (\ell-x) \Mb - \int_x^\ell \Abb_{\Phib(s)-\Phib(0)} ds \Nb.
$$
Now from the fourth equation we obtain
$$
\aligned
-\yb(x) &= \yb(\ell) - \yb(x) = \int_x^\ell \lab(s) -\tb(s) \times \thb(s) ds = \int_x^\ell \lab(s)ds -\int_x^\ell\Phib'(s) \times \thb(s) ds \\
&= \int_x^\ell \lab(s)ds +\int_x^\ell\Phib(s) \times \thb'(s) ds   + \Phib(x) \times \thb(x)\\
&= \int_x^\ell \lab(s)ds +\int_x^\ell\Phib(s) \times \mb(s) ds   + \Phib(x) \times \thb(x)\\
&= \int_x^\ell \lab(s)ds +\int_x^\ell\Phib(s) \times \left( \Mb - \Abb_{\Phib(s)-\Phib(0)} \Nb \right)ds   + \Phib(x) \times \thb(x)\\
&= \int_x^\ell \lab(s)ds +\int_x^\ell \Abb_{\Phib(s)} ds \Mb - \int_x^\ell \Abb_{\Phib(s)}  \Abb_{\Phib(s)-\Phib(0)} ds  \Nb   + \Phib(x) \times \thb(x).
\endaligned
$$
Applying the boundary conditions at $x=0$, for $\yb,\thb$ we obtain the equations
$$
\aligned
& 0 = \ell \Mb - \int_0^\ell \Abb_{\Phib(s)-\Phib(0)} ds \Nb,\\
& 0 = \int_0^\ell \lab(s)ds +\int_0^\ell \Abb_{\Phib(s)} ds \Mb - \int_0^\ell \Abb_{\Phib(s)}  \Abb_{\Phib(s)-\Phib(0)} ds  \Nb.
\endaligned
$$
Multiplying the first equation by $\Abb_{-\Phib(0)}$ and adding to the second we obtain
$$
0 = \int_0^\ell \lab(s)ds +\int_0^\ell \Abb_{\Phib(s)-\Phib(0)} ds \Mb - \int_0^\ell \Abb_{\Phib(s)-\Phib(0)}  \Abb_{\Phib(s)-\Phib(0)} ds  \Nb.
$$
Thus the system is given by
\begin{equation}\label{system6}
\left[\begin{array}{cc}
\ell \Ibb & - \int_0^\ell \Abb_{\Phib(s)-\Phib(0)} ds \\
\int_0^\ell \Abb_{\Phib(s)-\Phib(0)} ds & -\int_0^\ell \Abb_{\Phib(s)-\Phib(0)} \Abb_{\Phib(s)-\Phib(0)} ds
\end{array}\right] \left[ \begin{array}{c}
\Mb\\
\Nb
\end{array}\right] =
\left[ \begin{array}{c}
0\\
- \int_0^\ell \lab(s) ds
\end{array}\right].
\end{equation}
The matrix of the system, denote it by $\Mbbbb$, is symmetric. Moreover, for all $\Mb,\Nb \in \ZR^3$
$$
\aligned
\Mbbbb \left[ \begin{array}{c}
\Mb\\
\Nb
\end{array}\right] \times \left[ \begin{array}{c}
\Mb\\
\Nb
\end{array}\right] &= \ell \Mb \cdot \Mb - 2\int_0^\ell \Abb_{\Phib(s)-\Phib(0)} ds \Nb \cdot \Mb -\int_0^\ell \Abb_{\Phib(s)-\Phib(0)} \Abb_{\Phib(s)-\Phib(0)} ds \Nb \cdot \Nb\\
&=\int_0^\ell \left( \Mb - \Abb_{\Phib(s)-\Phib(0)} \Nb\right) \cdot \left( \Mb - \Abb_{\Phib(s)-\Phib(0)} \Nb\right)ds \geq 0
\endaligned
$$
and thus is positive semidefinite. It is positive definite unless, for some $0 \neq (\Mb,\Nb) \in \ZR^6$
$$
0 = \int_0^\ell \left( \Mb - \Abb_{\Phib(s)-\Phib(0)} \Nb\right) \cdot \left( \Mb - \Abb_{\Phib(s)-\Phib(0)} \Nb\right)ds.
$$
Then
$$
\Mb - \Abb_{\Phib(s)-\Phib(0)} \Nb = 0, \qquad s \in [0,\ell].
$$
Choosing $s=0$ we obtain $\Mb=0$. Therefore
$$
\Abb_{\Phib(s)-\Phib(0)} \Nb = 0, \qquad s \in [0,\ell].
$$
If $\Phib$ is not affine (the rod is not straight) this implies $\Nb=0$ and therefore $\Mbbbb$ is regular. Hence we have the unique solution of (\ref{system}). Therefore $\Im{B} = Q'= L^2(0,\ell)$.

If $\Phib$ is affine, i.e., $\Phib(s) -\Phib(0)= s \tb$, for some constant vector $\tb$ then $\Mbbbb$ is of rank $5$ with $\Ker{\Mbbbb} = \sspan\{(0,\tb)\}$ and with the image $\Im{\Mbbbb} = \{(0,\tb)\}^\perp$. Thus the system (\ref{system6}) has solution if and only if
\begin{equation}\label{uvjet}
\int_0^\ell \lab(s) ds \cdot \tb = 0.
\end{equation}
Therefore the system (\ref{system}) has a solution if and only if  (\ref{uvjet}) holds.
Therefore $\{\lab \in L^2(0,\ell;\ZR^3): \int_0^\ell \lab (s) ds \cdot \tb =0\} \subseteq \Im{B}$.
On the other hand for $\lab \in \Im{B}$ one has
$$
\yb' + \tb \times \thb = \lab.
$$
Integrating over $[0,\ell]$ we obtain
$$
\tb \times \int_0^\ell\thb(s) ds = \int_0^\ell \lab(s) ds.
$$
Thus (\ref{uvjet}) is fulfilled, $B(V^0) = \{\lab \in L^2(0,\ell;\ZR^3): \int_0^\ell \lab (s) ds \cdot \tb =0\}$ and thus $\Im B$ is closed (moreover of codimension 1).
\end{proof}

%

\begin{proposition}\label{pexistence_one}
\begin{itemize}
\item[a)] For every $\fb \in L^2(0,\ell; \ZR^3)$ the problem (\ref{weak_one_0}) has a unique solution $(\yb,\thetab) \in V_0$.

\item[b)]
For every $\fb \in L^2(0,\ell; \ZR^3)$ the problem (\ref{weak_one}) has a unique solution $((\yb,\thetab),\nb)\in V\times Q$. The function $(\yb,\thetab)$ also satisfies (\ref{weak_one_0}).

\item[c)] If $(\yb,\thb)\in V_0$ is the solution of (\ref{weak_one_0}) then there is $\nb\in Q$ such that $((\yb,\thb),\nb)$ is solution of (\ref{weak_one}).
\end{itemize}
\end{proposition}
\begin{proof}
The statement a) is a direct consequence of Lemma~\ref{lexistence_one}a), continuity of forms $k$ and $b$ and linear functional $l$ and the Lax--Milgram lemma.

Statements b) and c) are classical results about the linear variational problems with constraints, see \cite[Theorem~4.1 and Corollary 4.1]{GR}
or \cite[Theorem II.1.1]{BrezziFortin}.
\end{proof}

\section{Stent as a 3D net of 1D curved rods\label{S3}}\label{Sstent}
\setcounter{equation}{0}

\subsection{Differential formulation}

As mentioned earlier, stent is a three-dimensional elastic body defined as a union of three-dimensional struts.
\begin{figure}[ht]
\begin{center}
\includegraphics[width=0.65\textwidth]{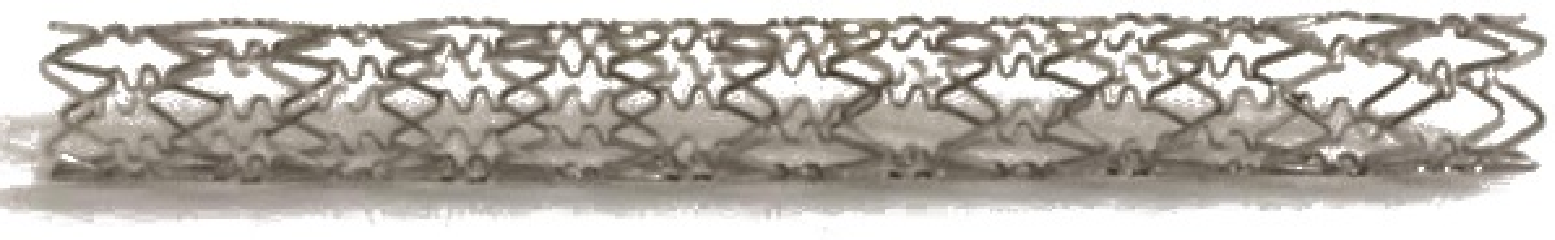}\\
\includegraphics[width=0.65\textwidth]{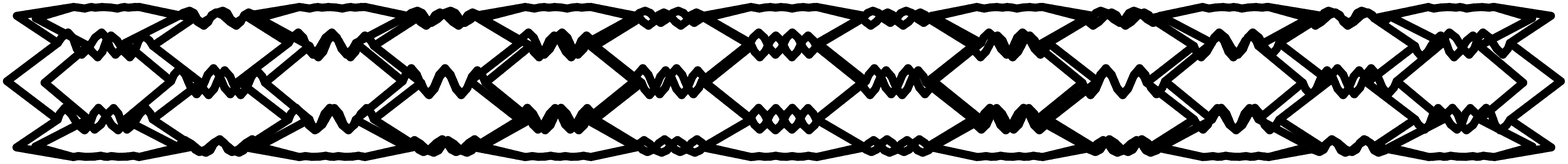}
\end{center}
\caption{Cypher stent by Cordis (upper figure) and its 1d computer idealization (lower figure)\label{cypher}}
\end{figure}
Each strut we model by the one-dimensional curved rod model.
Thus for defining the one-dimensional model we only need to prescribe:
\begin{itemize}
\item $\cV$ set of $n_\cV$ vertices of the stent (points where middle lines meet),
\item $\cN$ set of $n_\cN$ edges of the stent (pairing of vertices),
\item $\Phib^i:[0,\ell_i] \to \ZR^3$ parametrization of the middle line of $i$th strut (edge $e_i \in \cN$), $i=1,\ldots,n_\cN$,
\item $\mu_i,E_i$ parameters of material from which $i$th strut is made of, $i=1,\ldots,n_\cN$,
\item $I^i_{\alpha,\beta}$ $\alpha, \beta=1,2$ and $K^i$ moments of inertia and torsional rigidity of cross-sections of $i$th strut, $i=1,\ldots,n_\cN$.
\end{itemize}
Note that $(\cV, \cN)$ defines a graph and sets the topology of the stent. Adding precise geometry of struts by prescribing parametrizations is also important. This actually introduces orientation in the
graph however it is not important for the mechanics of the system.

A one-dimensional model of the given three-dimensional stent is given by the family of equations on each strut (edge)
\begin{eqnarray}
\label{lj1ri} &&0 = \partial_s \nb^i + \fb^i,\\
\label{lj2ri} &&0 = \partial_s \mb^i  +  \tb^i \times\nb^i%
,\\
\label{lj3ri} &&0 = \partial_s \thetab^i - \Qbb^i (\Hbb^i)^{-1} (\Qbb^i)^T \mb^i,\\
\label{lj4ri}
&&0 = \partial_s \yb^i +\tb^i \times  \thetab^i
\end{eqnarray}
for all $e_i \in \cN$. Additionally,
%
we need to prescribe the coupling conditions that need to be satisfied at each vertex of the stent net
where the edges (stent struts) meet.
As mentioned earlier, two sets of coupling conditions hold:
\begin{itemize}
\item the kinematic coupling condition: $(\yb,\thetab)$ continuous at each vertex,
\item the dynamic coupling condition: balance of contact forces ($\nb$) and contact moments ($\mb$) at each vertex,
\end{itemize}
\begin{equation}\label{ccs}
\aligned
&\sum_{i\in J^+_j} \nb^i(\ell_i) - \sum_{i\in J^-_j} \nb^i(0) = 0, \qquad j=1,\ldots, n_\cV,\\
&\sum_{i\in J^+_j} \mb^i(\ell_i) - \sum_{i\in J^-_j} \mb^i(0) = 0, \qquad j=1,\ldots, n_\cV,\\
&
\thetab^i(0) = \thetab^k(\ell^k), \qquad i \in J^-_j, k\in J^+_j, \qquad j=1,\ldots, n_\cV,\\
&
\yb^i(0) = \yb^k(\ell^k), \qquad i \in J^-_j, k\in J^+_j, \qquad j=1,\ldots, n_\cV;
\endaligned
\end{equation}
here $J^-_j$ stands for the set of all edges that leave (i.e. the local variable is  equal 0 at) the vertex $j$ and $J^+_j$ stands for the set of all edges that enter (i.e. the local variable is  equal $\ell$ at) the vertex $j$.

This constitutes the one-dimensional model of stents. Since the body is not fixed at any point the solution is not unique and there is associated necessary condition for the existence. It is easy to check that the functions
\begin{equation}\label{ker}
\yb^i (s)= \const_y - \Phib^i (s) \times \const_\theta, \qquad  \thb^i(s) = \const_\theta, \qquad i =1, \ldots, n_\cN
\end{equation}
satisfy (\ref{lj1ri})--(\ref{lj4ri}) for zero loads ($\fb^i=0$, $i=1,\ldots,n_\cN$). Further, to obtain the problem with unique solution we add the conditions
\begin{equation}\label{uniqueness}
 \sum_{i=1}^{n_\cN} \int_0^{\ell_i} \yb^i ds =  \sum_{i=1}^{n_\cN} \int_0^{\ell_i} \thb^i ds =0.
\end{equation}
Accompanied with this nonuniqueness are necessary conditions for the existence
\begin{equation}\label{necessary}
\sum_{i=1}^{n_\cN} \int_0^{\ell_i}\fb^i (s) ds = 0, \qquad  \sum_{i=1}^{n_\cN} \int_0^{\ell_i}\Phib^i(s) \times \fb^i (s) ds =0.
\end{equation}
These conditions are exactly the equilibrium conditions for the total force and the total moment.
For more details in a bit more complex setting see \cite{Zugec0, Zugec1}.

Thus the model of stents we consider in the sequel is given by the rod equations (\ref{lj1ri})--(\ref{lj4ri}), contact conditions (\ref{ccs}) and the conditions of zero mean displacement and zero mean rotation (\ref{uniqueness}).

\subsection{Weak and mixed formulations}

Next we turn to the weak and mixed formulation of the problem.
The kinematic coupling conditions are satisfied by including this condition into the space of test functions, thereby requiring that all possible candidates
for the solution must satisfy the continuity of displacement and
the continuity of infinitesimal rotation at every net vertex
(avoiding the stent rupture (caused by jump in displacements or infinitesimal rotations of the cross-section), in which case the model equations cease to be valid).
We begin by first defining the space of $H^1$-functions $\ub$, defined on the entire stent net $\cN$, such that
they satisfy the
kinematic coupling condition at each vertex $\Vb\in\cV$.
The vector function $\ub$ consist of all the state variables $(\yb,\thetab)$ defined on all the edges $\eb^i, i = 1,...,n_\cN$,
so that
$$
\ub = (\ub^1,...,\ub^{n_\cN}) = ((\yb^1,\thetab^1),...,(\yb^{n_\cN},\thetab^{n_\cN})).
$$
The kinematic coupling condition requires that the displacement of the middle line $\yb$, and the infinitesimal rotation of the cross-section $\thetab$,
are continuous at every vertex $\Vb\in\cV$.
More precisely, at each vertex $\Vb\in\cV$ at which the edges $\eb^i$ and $\eb^j$ meet,
the kinematic condition says that the trace of $\ub^i$ evaluated at the value of the parameter $s \in \{0,\ell_i\}$
that corresponds to the vertex $\Vb$, i.e., $\ub^i((\Phib^i)^{-1}(\Vb))$, has to be equal to the trace
$\ub^j((\Phib^j)^{-1}(\Vb))$.
Thus, for $k\in \ZN$, we define the space
$$
\aligned
H^1(\cN; \ZR^k) = \bigg\{ \ub&=(\ub^1, \ldots, \ub^{n_\cN}) \in \prod_{i=1}^{n_\cN} H^1(0,\ell_i; \ZR^k) : \\
& \ub^i((\Phib^i)^{-1}(\Vb)) = \ub^j((\Phib^j)^{-1}(\Vb)), \
\forall \Vb \in \cV, \Vb \in \eb^i \cap \eb^j \bigg\}.
\endaligned
$$

The dynamic coupling conditions, however,
are satisfied in the weak sense by imposing this condition in the weak formulation of the underlying equations.
To get to the weak formulation of the mixed formulation of the stent problem we sum up the weak forms of the mixed formulations for each strut on a
space of functions $V_S= H^1(\cN;\ZR^6)$ which are defined on the whole stent and which are continuous in vertices (globally continuous). Then by the dynamic contact
conditions contact couples and forces from the right hand side of (\ref{weak_one}) cancel out.
Let us denote by
$$
V_S = H^1(\cN;\ZR^6), \qquad Q_S = L^2(\cN;\ZR^3) \times \ZR^3 \times \ZR^3= \prod_{i=1}^{n_\cN} L^2 (0,\ell_i; \ZR^3) \times \ZR^3 \times \ZR^3
$$
the function spaces and define bilinear forms obtained by summing the associated forms for each rod
\begin{equation}\label{forme_stent}
\aligned
k_S: V_S \times V_S \to \ZR, \qquad k_S(\ub_S,\utb_S)=&\ \sum_{i=1}^{n_\cN}\int_0^{\ell_i} \Qbb^i \Hbb^i (\Qbb^i)^T \partial_s \thetab^i  \cdot \partial_s  \thetatb^i ds,\\
b_S: Q_S \times V_S \to \ZR, \qquad b_S(\nb_S,\utb_S)=&\ \sum_{i=1}^{n_\cN}\int_0^{\ell_i}  \nb^i \cdot (\partial_s \ytb^i + \tb^i \times \thetatb^i ) ds\\
&\ + \alb \cdot \sum_{i=1}^{n_\cN} \int_0^{\ell_i} \ytb^i ds+ \betab \cdot \sum_{i=1}^{n_\cN} \int_0^{\ell_i} \thetatb^i ds
\endaligned
\end{equation}
and the linear functional
$$
l_S:V_S \to \ZR, \qquad l_S(\utb_S) = \sum_{i=1}^{n_\cN} \int_0^{\ell_i} \fb^i \cdot \ytb^i ds;
$$
here, we use the notation
$$
\nb_S = (\nb^1, \ldots, \nb^{n_\cN}, \alb,\betab).
$$
Let us also define the function space
$$
\Vstent = \{\utb_S\in V_S : b_S(\ntb_S,\ub_S)=0,\ntb_S\in Q_S\}.
$$
Now the weak formulation is given by:
find $\ub_S\in \Vstent$ such that
\begin{equation}\label{weak_stent_0}
k_S(\ub_S,\utb_S) = l_S(\utb_S), \quad \utb_S \in \Vstent.
\end{equation}
More details on the model can be found in \cite{IMASJ}. This model actually is not limited for stents. It can be used to model any elastic structure made of rods. The associated static model is rigorously justified in \cite{Griso} from three-dimensional linearized elasticity.
The weak formulation of the problem is essential for obtaining the numerical approximation using the finite element method.
However because of the inextensiblity and unshearability constraints that are difficult to satisfy by the finite elements the mixed problem is important.

%
%
%
%

The mixed formulation of the problem (\ref{weak_stent_0}) is given by: find $(\ub_S,\nb_S) \in V_S \times Q_S$ such that
\begin{equation}\label{mixed_stent}
\aligned
& k_S(\ub_S,\utb_S) + b_S(\nb_S,\utb_S) = l_S(\utb_S) , \quad \utb_S \in V_S,\\
&b_S (\ntb_S,\ub_S)=0,\qquad \ntb_S \in Q_S.
\endaligned
\end{equation}

In the sequel we analyze the existence and uniqueness and the relation of the solution of the weak and mixed formulations.
Obvious part is that any solution $(\ub_S,\nb_S) \in V_S \times Q_S$ of the mixed formulation is the solution of the weak formulation. For the opposite some more analysis has to be made.

Equivalence of the weak formulation and the mixed formulation is important in order to go back to the strong formulation. Then from the mixed formulation it is easy to conclude the regularity result, namely
$$
\nb^i \in H^1(0,\ell_i), \quad \mb^i \in H^2(0,\ell_i), \quad \thb^i \in H^3(0,\ell_i), \quad \ub^i \in H^4(0,\ell_i).
$$
Even more smoothness is obtained if the force density is assumed more regular on edges.
Further, the strong formulation is then given by: find  $(\yb^i, \thetab^i, \mb^i, \nb^i)$, $i=1,\ldots,n_\cN$ and $\alb, \betab \in \ZR^3$ that satisfy the equations at edges
\begin{equation}\label{stent}
\aligned
&\partial_s \nb^i - \alb +\fb^i =0, \qquad i=1,\ldots,n_\cN,\\
&\partial_s \mb^i + \tb^i \times \nb^i - \betab =0, \qquad i=1,\ldots,n_\cN,\\
&\partial_s \thetab^i - \Qbb^i (\Hbb^i)^{-1} (\Qbb^i)^T \mb^i =0, \qquad i=1,\ldots,n_\cN,\\
&\partial_s \yb^i + \tb^i \times \thetab^i =0, \qquad i=1,\ldots,n_\cN,
\endaligned
\end{equation}
that has total mean displacement and rotation zero, i.e. (\ref{uniqueness}) holds
and that the kinematic and dynamic contact conditions (\ref{ccs}) at all vertices hold.

\begin{remark}\label{rab}\em
Inserting the displacements and infinitesimal rotations from the kernel of the stent operator, i.e. of the form (\ref{ker}), in the mixed formulation (\ref{mixed_stent}) we obtain the equation
$$
\alb \cdot \sum_{i=1}^{n_\cN} \int_0^{\ell_i} (\const_y - \Phib^i (s) \times \const_\theta) ds+ \betab \cdot \sum_{i=1}^{n_\cN} \int_0^{\ell_i} \const_\theta ds =  \sum_{i=1}^{n_\cN} \int_0^{\ell_i} \fb^i \cdot (\const_y - \Phib^i (s) \times \const_\theta) ds,
$$
for all $\const_y, \const_\theta \in \ZR^3$. This implies
$$
\aligned
&\alb   =  \sum_{i=1}^{n_\cN} \int_0^{\ell_i} \fb^i ds \bigg\slash \sum_{i=1}^{n_\cN} \ell_i ,\\
& \betab  =  \left(-\sum_{i=1}^{n_\cN} \int_0^{\ell_i} \fb^i \times \Phib^i (s)  ds + \alb \times \sum_{i=1}^{n_\cN} \int_0^{\ell_i}  \Phib^i (s)  ds\right)\bigg\slash \sum_{i=1}^{n_\cN}\ell_i.
\endaligned
$$
Note that the weak formulation (\ref{weak_stent_0}) will have a unique solution as a consequence of Lemma~\ref{lexistence_stent}, since the form $k_S$ is $\Vstent$--elliptic. Also, there is no necessary condition for the existence. The role here has been done by the multipliers $\alb$ and $\betab$ since they are chosen such that $\fb^i-\alb$ satisfy the necessary conditions of the form (\ref{necessary}).
Note further that if we take loads that satisfy necessary condition (\ref{necessary}) then $\alb=\betab=0$.
\end{remark}

First we prove the Poincare type estimate on the stent.
\begin{lemma}\label{lPoincare}
There is $C>0$ such that
$$\aligned
\sum_{i=1}^{n_\cN} \|\yb^i\|^2_{H^1(0,\ell_i;\ZR^3)} &\leq C_P \left(\sum_{i=1}^{n_\cN} \|{\yb^i}'\|^2_{L^2(0,\ell_i;\ZR^3)}
+ |\sum_{i=1}^{n_\cN}\int_0^{\ell_i} \yb^i|^2\right),\\
\yb_S &= (\yb^1, \ldots,\yb^{n_\cN}) \in H^1(\cN;\ZR^3).
\endaligned$$
\end{lemma}
\begin{proof}
Let us suppose that the estimate does not hold. Then there is a sequence $(\yb_{Sn}) \subset H^1(\cN;\ZR^3)$ such that
\begin{equation}\label{ctd}
\sum_{i=1}^{n_\cN} \|\yb^i_n\|^2_{H^1(0,\ell_i;\ZR^3)} = 1
\end{equation}
and that
\begin{equation}\label{kvg}
\aligned
&\yb^i_n \rightharpoonup \yb^i \qquad \mbox{ weakly in } H^1(0,\ell_i; \ZR^3),\qquad i=1,\ldots, n_\cN,\\
&{\yb^i_n}' \to 0 \qquad \mbox{ strongly in } L^2(0,\ell_i; \ZR^3),\qquad i=1,\ldots, n_\cN,\\
& \sum_{i=1}^{n_\cN}\int_0^{\ell_i} \yb^i_n \to 0.
\endaligned
\end{equation}
From the first and second convergence ${\yb^i}'=0$ and thus $\yb^i = \const, i=1, \ldots, n_\cN$ due to the continuity requirement
from the definition of $H^1(\cN;\ZR^3)$. From the first and third convergence in (\ref{kvg}) we obtain that
$$
\sum_{i=1}^{n_\cN}\int_0^{\ell_i} \yb^i =0.
$$
This now implies $\yb_S=0$. The first and second convergence in (\ref{kvg}) also imply $\yb^i_n \to \yb^i$ strongly in $H^1(0,\ell_i; \ZR^3)$ for $i=1,\ldots,n_\cN$, i.e., $\yb^i_n \to 0$ strongly in $H^1(0,\ell_i;\ZR^3)$,
which is in contradiction with (\ref{ctd}).
\end{proof}

\begin{lemma}\label{lexistence_stent}
The form $k_S$ is $\Vstent$--elliptic.
\end{lemma}
\begin{proof}
As in the proof of Lemma~\ref{lexistence_one} we estimate the function $\ub_S \in \Vstent$. In the first estimate we use Lemma~\ref{lPoincare} and (\ref{uniqueness})
$$
\aligned
&\hspace{-1.5ex}\|\ub_S\|^2_{H^1(\cN; \ZR^6)}= \sum_{i=1}^{n_\cN} (\|\yb^i\|_{H^1(0,\ell_i; \ZR^3)}^2 + \|\thetab^i\|_{H^1(0,\ell_i; \ZR^3)}^2 )\\
&\leq C_P(\sum_{i=1}^{n_\cN} (\|{\yb^i}'\|_{L^2(0,\ell_i; \ZR^3)}^2  + \|{\thetab^i}'\|_{L^2(0,\ell_i; \ZR^3)}^2) + |\sum_{i=1}^{n_\cN}\int_0^{\ell_i} \yb^i|^2 +
|\sum_{i=1}^{n_\cN}\int_0^{\ell_i} \thetab^i|^2)\\
&\leq 2C_P \!\!\sum_{i=1}^{n_\cN} (\|{\yb^i}'\!+\tb^i \!\times\! \thetab^i\|_{L^2(0,\ell_i; \ZR^3)}^2 + \|\tb^i \!\times\! \thetab^i\|_{L^2(0,\ell_i; \ZR^3)}^2+ \|{\thetab^i}'\|_{L^2(0,\ell_i; \ZR^3)}^2)\\
&= C \sum_{i=1}^{n_\cN} \|{\thetab^i}'\|_{L^2(0,\ell_i; \ZR^3)}^2\\
&\leq \frac{C}{\min_i \sigma(\Hbb^i)} k_S(\ub_S,\ub_S).
\endaligned
$$
\end{proof}

Together with continuity of the forms $k_S$, $b_S$ and the linear functional $l_S$ this Lemma implies the existence theorem for the weak formulation (\ref{weak_stent_0}).
We proceed with the analysis of the mixed formulation only for stents which belong to the class $\cS$ (see below) since then we know how to prove the $\inf\sup$ estimates (see Lemma~\ref{infsup} for more details).

\begin{definition}\label{classS}
The stent belongs to the class $\cS$ if one of the following is satisfied
\begin{itemize}
\item all edges are curved,
\item there are straight edges. Then
$$
\sum_{i\in J^+_j} \alpha_i \tb^i - \sum_{i\in J^-_j} \alpha_i \tb^i = 0, \quad j=1,\ldots, n_\cV, \quad \Leftrightarrow \quad \alpha_i=0, \quad i=1,\ldots,n_\cN;
$$
here $\alpha_i=0$ for edges which are not straight.
\end{itemize}
\end{definition}

\begin{lemma}\label{infsup}
Let the stent be in the class \cS. Then there is $\beta_{BB}>0$ such that
$$
\inf_{\ntb_S \in Q_S} \sup_{\utb_S \in V_S} \frac{b_S(\ntb_S,\utb_S)}{\|\ntb_S\|_{Q_S} \|\utb_S\|_{H^1(\cN;\ZR^6)}} \geq \beta_{BB}.
$$
\end{lemma}
\begin{proof}
For a given $\ntb_S=(\lab_1,\ldots,\lab_{n_\cN},\alb,\betab) \in Q_S$ we will find $\ub_S = ((\yb^1,\thetab^1), \ldots,(\yb^{n_\cN},\thetab^{n_\cN})) \in V_S$ such that
\begin{equation}\label{lambda}
\aligned
&\partial_s \yb^i + \tb^i \times \thetab^i =\lab^i, \qquad i=1,\ldots,n_\cN,\\
&\sum_{i=1}^{n_\cN}\int_0^{\ell_i} \yb^i ds = \alb,\\
&\sum_{i=1}^{n_\cN}\int_0^{\ell_i} \thetab^i ds = \betab,
\endaligned
\end{equation}
and such that there is a constant $C$ independent of $\ub_S$ and $\ntb_S$ for which
\begin{equation}\label{estest}
\|\ub_S\|_{H^1(\cN;\ZR^3)} \leq C \|\ntb_S\|_{Q_S}.
\end{equation}
The statement of the lemma then follows since
$$
\sup_{\utb_S \in V_S}\frac{b_S(\ntb_S,\utb_S)}{\|\ntb_S\|_{Q_S} \|\utb_S|_{H^1(\cN;\ZR^6)}} \geq \frac{\|\ntb_S\|_{Q_S}^2}{\|\ntb_S\|_{Q_S} \|\ub_S\|_{H^1(\cN;\ZR^6)}} =\frac{\|\ntb_S\|_{Q_S}}{\|\ub_S\|_{H^1(\cN;\ZR^6)}}  \geq  \frac{1}{C} =: \beta_{BB}.
$$

We impose more restrictions that still lead us to the solution of (\ref{lambda}):
\begin{equation}\label{lambda1}
\aligned
&\partial_s \nb^i  =0, \qquad i=1,\ldots,n_\cN,\\
&\partial_s \mb^i + \tb^i \times \nb^i =0, \qquad i=1,\ldots,n_\cN,\\
&\partial_s \thetab^i - \mb^i =0, \qquad i=1,\ldots,n_\cN,\\
&\partial_s \yb^i + \tb^i \times \thetab^i =\lab^i, \qquad i=1,\ldots,n_\cN,\\
&\sum_{i=1}^{n_\cN}\int_0^{\ell_i} \yb^i ds = \alb,\\
&\sum_{i=1}^{n_\cN}\int_0^{\ell_i} \thetab^i ds = \betab,
\endaligned
\end{equation}
and the functions $(\yb^i, \thetab^i, \mb^i, \nb^i)$, $i=1,\ldots,n_\cN$ have to satisfy the kinematic and dynamic contact conditions:
\begin{equation}\label{cc}
\aligned
&\sum_{i\in J^+_j} \nb^i(\ell_i) - \sum_{i\in J^-_j} \nb^i(0) = 0, \qquad j=1,\ldots, n_\cV,\\
&\sum_{i\in J^+_j} \mb^i(\ell_i) - \sum_{i\in J^-_j} \mb^i(0) = 0, \qquad j=1,\ldots, n_\cV,\\
&\Thetab^j = \thetab^i(0) = \thetab^k(\ell^k), \qquad i \in J^-_j, k\in J^+_j, \qquad j=1,\ldots, n_\cV,\\
&\Yb^j = \yb^i(0) = \yb^k(\ell^k), \qquad i \in J^-_j, k\in J^+_j, \qquad j=1,\ldots, n_\cV.
\endaligned
\end{equation}
This system corresponds to the equilibrium stent problem with zero forcing, for specific material and with  the proposed extension given by $\lab^i$'s. This problem is very similar to  (\ref{lj1ri})--(\ref{ccs}).

From the first equation in (\ref{lambda1}) we conclude that all $\nb^i$ are constant. Integrating the second equation we obtain that
$$
\mb^i(s) = \mb^i(\ell_i) + \int_s^{\ell_i} \tb^i(r) dr \times \nb^i = \mb^i(\ell_i) + \Phitb^i(s)\times \nb^i,
$$
where $\Phitb^i(s) = \Phib^i(\ell_i) - \Phib^i(s)$.
Let us now insert the values for $\nb^i$ and $\mb^i$ to the dynamical contact conditions (the first two in (\ref{cc})). We obtain
\begin{equation}\label{cc1}
\aligned
&\sum_{i\in J^+_j} \nb^i - \sum_{i\in J^-_j} \nb^i = 0, \qquad j=1,\ldots, n_\cV,\\
&\sum_{i\in J^+_j} \mb^i(\ell_i)  - \sum_{i\in J^-_j} (\mb^i(\ell_i) + \Phitb^i(0) \times \nb^i)= 0, \qquad j=1,\ldots, n_\cV.
\endaligned
\end{equation}
Let $\Abb_\cI \in M_{3n_\cV, 3n_\cN}(\ZR)$ denote the incidence matrix of the oriented graph $(\cV,\cN)$ with three connected components (organized in the following way: a $3\times3$ submatrix at rows $3i-2,3i-1,3i$ and columns $3j-2,3j-1,3j$ is $\Ibb$ if the edge $j$ enters the vertex $i$, $-\Ibb$ if it leaves the vertex $i$ or $0$ otherwise). Let us also denote projectors
$$
\Pbbb^i_\cN \in M_{3,3n_\cN}, \qquad \Pbbb^j_\cV\in M_{3,3n_\cV}
$$
on the coordinates $3i-2,3i-1,3i$ and $3j-2,3j-1,3j$, respectively. Then we define the matrix
\begin{equation}\label{ap}
\Abb_\Phib = \sum_{j=1}^{n_\cV} \sum_{i\in J^-_j} (\Pbbb^j_\cV)^T \Abb_{\Phitb^i(0)} \Pbbb^i_\cN
\end{equation}
(the matrix $\Abb_\vb$ is the skew-symmetric matrix associated with the axial vector $\vb$, i.e., $\Abb_\vb \xb = \vb \times \xb, \xb\in \ZR^3$).
Using the notation
$$
\Nb = ((\nb^1)^T, \ldots, (\nb^{n_\cN})^T)^T, \qquad \Mb = ((\mb^1(\ell^1))^T, \ldots, (\mb^{n_\cN}(\ell^{n_\cN}))^T)^T
$$
since
$$
\Pbbb^j_\cV \Abb_\Phib \Nb = \sum_{i\in J^-_j} \Abb_{\Phitb^i(0)} \Pbbb^i_\cN \Nb = \sum_{i\in J^-_j} \Phitb^i(0) \times \nb^i
$$
the equations (\ref{cc1}) can be written by
\begin{equation}\label{eqNM}
-\Abb_\cI \Nb =0, \qquad -\Abb_\cI \Mb + \Abb_\Phib \Nb =0.
\end{equation}

The integration of the  third equation in (\ref{lambda1}) implies
\begin{equation}\label{thetas}
\thetab^i(s) - \thetab^i(0)= \int_0^{s} \mb^i(r) dr = s \mb^i(\ell_i) + \int_0^{s}\Phitb^i(r)\times \nb^i dr = s \Pbbb^i_\cN \Mb  + \int_0^{s}\Abb_{\Phitb^i(r)}dr \Pbbb^i_\cN \Nb.
\end{equation}
Integration of the fourth equation in (\ref{lambda1}) implies
$$
\aligned
\yb^i(\ell_i) -  \yb^i(0) &= - \int_0^{\ell_i} \tb^i(s) \times \thetab^i(s) ds  + \int_0^{\ell_i}\lab^i(s) ds\\
&= - \int_0^{\ell_i} (\Phib^i)'(s) \times \left(\thetab^i(0)+\int_0^s \mb^i(r) dr \right) ds  + \int_0^{\ell_i}\lab^i(s) ds\\
&= - \Phitb^i(0) \times \thetab^i(0) - \left(\Phib^i(s) \times \int_0^s \mb^i(r) dr \right) \bigg|_0^{\ell_i}+ \int_0^{\ell_i} \Phib^i(s) \times  \mb^i(s)  ds  + \int_0^{\ell_i}\lab^i(s) ds\\
&= - \Phitb^i(0) \times \thetab^i(0) - \Phib^i(\ell_i) \times \int_0^{\ell_i} \mb^i(s) ds + \int_0^{\ell_i} \Phib^i(s) \times  \mb^i(s)  ds  + \int_0^{\ell_i}\lab^i(s) ds\\
&= - \Phitb^i(0) \times \thetab^i(0) -  \int_0^{\ell_i} \Phitb^i(s) \times  \mb^i(s)  ds  + \int_0^{\ell_i}\lab^i(s) ds.
%
 \endaligned
$$
Therefore
\begin{equation}\label{*}
\yb^i(\ell_i) -  \yb^i(0)= - \Abb_{\Phitb^i(0)} \thetab^i(0) -  \int_0^{\ell_i} \Abb_{\Phitb^i(s)} (\Pbbb^i_\cN \Mb + \Abb_{\Phitb^i(s)} \Pbbb^i_\cN \Nb)  ds  + \int_0^{\ell_i}\lab^i(s) ds.
\end{equation}
Next we introduce three vectors
$$\aligned
\Thetab &= ((\Thetab^1)^T, \ldots, (\Thetab^{n_\cV})^T)^T, \\
 \Yb &= ((\Yb^1)^T, \ldots, (\Yb^{n_\cV})^T)^T, \\
 \Lab &= ((\int_0^{\ell_1}\lab^1(s)ds)^T, \ldots, (\int_0^{\ell_{n_\cN}}\lab^{n_\cN}(s)ds)^T)^T.
 \endaligned
$$
The equation (\ref{thetas}) for $s=\ell_i$ is now given by
$$
\Pbbb^i_\cN \Abb_\cI^T \Thetab = \ell_i \Pbbb^i_\cN \Mb  + \int_0^{\ell_i}\Abb_{\Phitb^i(s)}ds \Pbbb^i_\cN \Nb.
$$
Thus we obtain
$$
\Abb_\cI^T \Thetab = \sum_{i=1}^{n_\cN} (\Pbbb^i_\cN)^T\Pbbb^i_\cN \Abb_\cI^T \Thetab =  (\sum_{i=1}^{n_\cN} (\Pbbb^i_\cN)^T \ell_i \Pbbb^i_\cN) \Mb  +  (\sum_{i=1}^{n_\cN} (\Pbbb^i_\cN)^T \int_0^{\ell_i}\Abb_{\Phitb^i(s)}ds \Pbbb^i_\cN ) \Nb
$$
which gives the third equation
\begin{equation}\label{treca}
-\Abb_\cI^T \Thetab + \Nbb_K \Mb  +  \Nbb_{K\Abb_\Phitb} \Nb =0,
\end{equation}
where the matrices
$$
\Nbb_K = \sum_{i=1}^{n_\cN} (\Pbbb^i_\cN)^T \ell_i \Pbbb^i_\cN, \qquad \Nbb_{K\Abb_\Phitb} = \sum_{i=1}^{n_\cN} (\Pbbb^i_\cN)^T \int_0^{\ell_i}\Abb_{\Phitb^i(s)}ds \Pbbb^i_\cN
$$
are block diagonal matrices with diagonal elements given by $\ell_i \Ibb$ and $\int_0^{\ell_i}\Abb_{\Phitb^i(s)}ds$, respectively.

The last equation we obtain from the integration of the fourth equation, i.e. (\ref{*}). Let us use the notation: $\sigma(i) = $ numeration of the leaving vertex of the edge $i$. Then we obtain
$$
\aligned
\Pbbb^i_\cN \Abb_\cI^T \Yb &= - \Abb_{\Phitb^i(0)} \Pbbb^{\sigma(i)}_\cV \Thetab -  \int_0^{\ell_i} \Abb_{\Phitb^i(s)} ds\ \Pbbb^i_\cN \Mb - \int_0^{\ell_i} \Abb_{\Phitb^i(s)} \Abb_{\Phitb^i(s)} ds\  \Pbbb^i_\cN \Nb  + \Pbbb^i_\cN \Lab.
 \endaligned
$$
Therefore, similarly as before we obtain
\begin{equation}\label{cetvrta}
-\Abb_\cI^T \Yb  -  \Abb_\Qb \Thetab -  \Nbb_{K\Abb_\Phitb} \Mb - \Nbb_{\Abb_\Phitb K \Abb_\Phitb} \Nb  +  \Lab = 0,
\end{equation}
where
$$
\Nbb_{\Abb_\Phitb K \Abb_\Phitb} = \sum_{i=1}^{n_\cN} (\Pbbb^i_\cN)^T \int_0^{\ell_i} \Abb_{\Phitb^i(s)} \Abb_{\Phitb^i(s)} ds\  \Pbbb^i_\cN
$$
is the diagonal matrix with diagonal elements given by $\int_0^{\ell_i} \Abb_{\Phitb^i(s)} \Abb_{\Phitb^i(s)} ds$ and
$$
\Abb_\Qb = \sum_{i=1}^{n_\cN} (\Pbbb^i_\cN)^T \Abb_{\Phitb^i(0)} \Pbbb^{\sigma(i)}_\cV.
$$
Note that the sum in the definition of $\Abb_\Phib$ is over all exiting edges from all vertices. Therefore this sum can be written over all edges but for prescribed exiting vertex. Therefore $\Abb_\Qb = - \Abb_\Phib^T$!
Therefore the system given by (\ref{eqNM}), (\ref{treca}), (\ref{cetvrta}) for $(\Yb,\Thetab,\Mb,\Nb)$ can be written by
\begin{equation}\label{thesystem}
\left( \begin{array}{cc}
\Bbbb & \Abb^T\\
\Abb & 0
\end{array}\right) \left( \begin{array}{c}
\xb\\ \yb
\end{array}\right) = \left( \begin{array}{c}
\fb\\ 0
\end{array}\right)  ,
\end{equation}
where
$$
\aligned
&\Bbbb = \left( \begin{array}{cc}
\Nbb_K  & \Nbb_{K\Abb_\Phitb} \\
-  \Nbb_{K\Abb_\Phitb} & - \Nbb_{\Abb_\Phitb K \Abb_\Phitb}
\end{array}\right),
\qquad
\Abb = \left( \begin{array}{cc}
-\Abb_\cI & \Abb_\Phib\\
0 & -\Abb_\cI
\end{array}\right),\\
& \xb = \left( \begin{array}{c}
\Mb\\ \Nb
\end{array}\right),
\qquad
\yb = \left( \begin{array}{c}
\Thetab\\ \Yb
\end{array}\right),
\qquad
\fb = \left( \begin{array}{c}
0\\ \Lab
\end{array}\right).
\endaligned
$$

Following \cite{VeselicVeselic} we compute the null space of the matrix
$$
\Hbb=\left( \begin{array}{cc}
\Bbbb & \Abb^T\\
\Abb & 0
\end{array}\right)
$$
as
$$
\Ker(\Hbb)=\left\{\left( \begin{array}{c}
\xb\\ \yb
\end{array}\right)~:~ \xb\in\Ker(\Bbbb)\cap\Ker(\Abb),~\yb\in\Ker(\Abb^T)\right\}
$$
and since according to \cite[Lema 3.4]{Zugec0} we have for stents of class $\mathcal{S}$ that $\Ker(\Bbbb)\cap\Ker(\Abb)=\{0\}$ we see that the vector $\phib:=\left( \begin{array}{cc}
\fb^T& 0
\end{array}\right)^T$ is orthogonal to $\Ker(\Hbb)$. For Hermitian matrices $\Hbb$ this is equivalent to the statement that $\phib\in\Ran(\Hbb)$, and so the system \eqref{thesystem} has at least one solution $\chib$ such that $\Hbb\chib=\phib$.

Let now $\Hbb^+$ be the Moore-Penrose generalized inverse of $\Hbb$. It is the unique Hermitian matrix $\Hbb^+$ such that matrices $\Hbb\Hbb^+$ and $\Hbb^+\Hbb$ are both orthogonal projections onto $\Ran(\Hbb)$ and $\Ran(\Hbb^+)$ respectively.
Recall that a matrix is an orthogonal projection if it is Hermitian and idempotent.

Thus the vector
$$
\chib_0 : =
\Hbb^+\phib
$$
satisfies
$$
\Hbb\chib_0
=\Hbb\Hbb^+\phib=\phib,
$$
since  $\phib\in\Ran(\Hbb)$ implies $\Hbb\Hbb^+\phib=\phib$. Therefore $\chib_0$
is a particular solution of \eqref{thesystem} whose norm is controlled by $\|\phib\|=\|\fb\|$.
Thus for $\left(\,\xb^T\ \yb^T\,\right)^T := \chib_0$ there is a constant $C=\|\Hbb^+\|$, depending only on the geometry of the stent such that
$$
\|\xb\| \leq C\|\fb\| \leq C\|\Lab\|, \qquad \|\yb\|\leq C \|\fb\| \leq C \|\Lab\|.
$$
Using the definition of $\xb, \yb$ and $\Lab$ we obtain
\begin{equation}\label{est0}
\|\Nb\|^2 + \|\Mb\|^2 + \|\Thetab\|^2 + \|\Yb\|^2\leq C  \sum_{i=1}^{n_\cN} \|\lab^i\|_{L^2(0,\ell_i;\ZR^3)}^2.
\end{equation}
These constants ($\Yb,\Thetab, \Mb, \Nb$) uniquely determine the function $\ub_S$ by (\ref{thetas}) and (\ref{*}) and for this solution one has
\begin{equation}\label{est0a}
\|\ub_S\|^2_{H^1(\cN;\ZR^6)}\leq C  \sum_{i=1}^{n_\cN} \|\lab^i\|_{L^2(0,\ell_i;\ZR^3)}^2.
\end{equation}

Next we need to satisfy the last two equations in (\ref{lambda}), so we define
$$
\aligned
&\Omegab = \frac{1}{\sum_i^{n_\cN}\ell_i} \left(\betab - \sum_{i=1}^{n_\cN} \int_0^{\ell_i}\thb^i (s)ds\right),\\
&\Ub =\frac{1}{\sum_i^{n_\cN}\ell_i} \left(\alb - \sum_{i=1}^{n_\cN} \int_0^{\ell_i}\yb^i (s)ds + \sum_{i=1}^{n_\cN} \int_0^{\ell_i}\Phib^i (s)ds \times \Omegab\right).
\endaligned
$$
Now we denote $\yhb^i = \yb^i + \Ub - \Phib^i\times \Omegab$, $\thhhb^i=\thb^i + \Omegab$. Then $(\yhb^i , \thhhb^i, \mb^i, \nb^i)$ satisfies the same equations as $(\yb^i, \thb^i, \mb^i, \nb^i)$, i.e., (\ref{lambda1}) and (\ref{cc}), but with different values in contacts, namely
$$
\yhb^i(0) = \Ub - \Phib^i(0) \times \Omegab, \quad \thhhb^i(0) = \Omegab, \quad \yhb^i(\ell_i) = \Ub - \Phib^i(\ell_i) \times \Omegab, \quad \thhhb^i(\ell_i) = \Omegab.
$$
However,
$$\ub^{\alb,\betab}_S = ((\yb^1+ \Ub - \Phib^1\times \Omegab, \thb^1+ \Omegab), \ldots, (\yb^{n_\cN}+ \Ub - \Phib^{n_\cN}\times \Omegab, \thb^{n_\cN}+ \Omegab))\in V_S
$$
and $\Ub$ and $\Omegab$ are defined such that
$$
\aligned
&\sum_{i=1}^{n_\cN} \int_0^{\ell_i} \yhb^i ds= \sum_{i=1}^{n_\cN} \int_0^{\ell_i} \yb^i ds +  \sum_{i=1}^{n_\cN}\ell_i \Ub -  \sum_{i=1}^{n_\cN}\int_0^{\ell_i} \Phib^i\times \Omegab ds = \alb,\\
&\sum_{i=1}^{n_\cN} \int_0^{\ell_i} \thhhb^i ds = \sum_{i=1}^{n_\cN} \int_0^{\ell_i} \thb^i ds +  \sum_{i=1}^{n_\cN} \ell_i \Omegab = \betab.
\endaligned
$$
Therefore from (\ref{est0a}) we obtain
\begin{equation}\label{est1}
\|\ub^{\alb,\betab}_S\|^2_{H^1(\cN;\ZR^6)}\leq C \left( \sum_{i=1}^{n_\cN} \|\lab^i\|_{L^2(0,\ell_i;\ZR^3)}^2 + \|\Ub\|^2 + \|\Omegab\|^2\right) \leq C \|\ntb_S\|^2_{Q_S}.
\end{equation}
Thus the function $\ub_S^{\alb,\betab}$ satisfies (\ref{lambda}) and (\ref{estest}) as announced and thus the lemma is proved.
\end{proof}

\begin{proposition}\label{pexistence_stent}
\begin{itemize}
\item[a)] The problem (\ref{weak_stent_0}) has a unique solution.
\item[b)] For every $l_S \in L^2(0,\ell; \ZR^3)'$ the problem (\ref{mixed_stent}) has a unique solution. This solution satisfies also the problem (\ref{weak_stent_0}).
\item[c)] Let $\ub_S\in \Vstent$ be the solution of (\ref{weak_stent_0}) then there is $\nb_S\in Q_S$ such that $(\ub_S,\nb_S)$ satisfies (\ref{mixed_stent}).
\item[d)] There exsits a constant $C$ such that
$$
\|\ub_S\|_{V_S}+\|\nb_S\|_{Q_S}\leq C_{LBB}\|l_S\|_{L^2(\cN; \ZR^3)}.
$$
\end{itemize}
\end{proposition}
\begin{proof}
As in Proposition~\ref{pexistence_one} the statement a) is a direct consequence of Lemma~\ref{lexistence_stent}a), continuity of forms $k_S$ and
$b_S$ and linear functional $l_S$ and the Lax--Milgram lemma.

Again, as in Proposition~\ref{pexistence_one} the statements b) and c) are classical results about the linear variational problems with constraints,
see \cite[Theorem~4.1 and Corollary 4.1]{GR} or \cite[Theorem II.1.1]{BrezziFortin}.
\end{proof}

%
%
%
%

\section*{Acknowledgement}
The research is supported by the Croatian Science Foundation grant nr. HRZZ 9345.

\end{document}